\documentclass[11pt]{article}
\usepackage[utf8]{inputenc}
\usepackage{float, amsmath, amsthm, amssymb, xcolor, commath, setspace, amsfonts, mathrsfs, graphicx, comment, blindtext, hyperref, lipsum, algorithm, algpseudocode, optidef, soul, authblk}
\usepackage[capitalise]{cleveref}
\usepackage[margin=1in]{geometry}
\newtheorem{theorem}{Theorem}[section]
\newtheorem{lemma}[theorem]{Lemma}

\DeclareMathOperator*{\argmin}{argmin}

\algrenewcommand\algorithmicrequire{\textbf{Input:}}
\algrenewcommand\algorithmicensure{\textbf{Output:}}
\newcommand{\RM}[1]{\textcolor{red}{\bf *** Ruta: #1 *** }}
\algrenewcommand{\Return}{\State\algorithmicreturn~}

\allowdisplaybreaks

\hypersetup{
    colorlinks=true,
    linkcolor=blue,
    filecolor=magenta,      
    urlcolor=cyan,
    pdftitle={Overleaf Example},
    pdfpagemode=FullScreen,
    }

\newtheorem{definition}{Definition}[section]
\newtheorem{assumption}{Assumption}[section]
\newtheorem{proposition}[theorem]{Proposition}
\newtheorem{corollary}{Corollary}[theorem]
\newtheorem{observation}[theorem]{Observation}

\newtheorem*{theorem*}{Theorem}

\theoremstyle{example}
\newtheorem*{remark}{Remark}

\theoremstyle{example}

\theoremstyle{example}

\newcommand{\ruta}[1]{\textcolor{red}{\bf *** #1 ***}}

\title{On the Envy-free Allocation of Chores}

\author{Lang Yin\qquad{}Ruta Mehta%
\vspace{1em}\\
University of Illinois Urbana-Champaign%
\vspace{0.0em}\\\normalsize
\texttt{\{\href{mailto:langyin2@illinois.edu}{langyin2},\href{mailto:rutameht@illinois.edu}{rutameht}\}@illinois.edu}%
}
\date{}

\begin{document}

\allowdisplaybreaks

\maketitle

\begin{abstract}
We study the problem of allocating a set of indivisible chores to three agents, among whom two have additive cost functions, in a fair manner. Two fairness notions under consideration are envy-freeness up to any chore (EFX) and a relaxed notion, namely envy-freeness up to transferring any chore (tEFX). In contrast to the case of goods, the case of chores remain relatively unexplored. In particular, our results constructively prove the existence of a tEFX allocation for three agents if two of them have additive cost functions and the ratio of their highest and lowest costs is bounded by two. In addition, if those two cost functions have identical ordering (IDO) on the costs of chores, then an EFX allocation exists even if the condition on the ratio bound is slightly relaxed. Throughout our entire framework, the third agent is unrestricted besides having a monotone cost function.
\end{abstract}
    

\section{Introduction}

The concept of fair division has been recorded in text earlier than two millennia ago in the Bible, where Abraham and Lot were planning to divide the land of Canaan. Since the 1940s, fair division has been prompted as among of the most fundamental problems in mathematical economics and game theory \cite{Steinhaus1948FairDivision} and has been studied extensively during the past decades, {\em e.g.,} see \cite{brams_taylor_1996, Moulin2004MIT, Aziz2016, Ghodsi2018, Caragiannis2016EFX, RePEc:arx:papers:1908.01669, meertens_2009}.

Fair division is the problem of allocating a set $M$ of $m$ items among a set $N$ of $n$ agents in a fair and efficient manner. The items may be beneficial (positively valued by the agents) like cake, cars, etc., in which case they are called goods, or may cause burden (negatively valued by agents) like teaching duties, in which case  they are called chores. In the case of goods, agents' preferences are defined by valuation function, say  $v_i: 2^M \rightarrow \mathbb{R}_{\geq 0}$ for agent $i\in N$, that the agents want to maximize. In the case of chores, agents' preferences are defined by a cost functions measuring the burden of chores, and denoted as $c_i: 2^M \rightarrow \mathbb{R}_{\geq 0}$ instead for agent $i$. Naturally, both $v_i$ and $c_i$ are monotone non-decreasing functions.

The {\em case of goods} has been extensively studied and is relatively batter understood than the {\em case of chores}, even though the latter is equally relevant in practice, for example dividing teaching duties among the faculty members, and sharing the cost for pollution control \cite{Traxler2002FairCD}. In this paper we study the fair division of indivisible chores with additive cost functions under one of the most sought after fairness notion of {\em envy-freeness up to one item (EFX)}. 

{\em Envy-freeness (EF)} is one of the fundamental fairness notions withing economics and social choice theory \cite{Budish2011EF1, Halevi2018, robertson_webb_1998, SEGALHALEVI20171, SEGALHALEVI201819, HaleviAndSziklai2019} (the other being {\em proportionality}) requiring that no agent should {\em envy} other's bundle over their own. Since EF is unachievable with indivisible items, for example allocating one iPhone among two agents who value it highly, various relaxation of EF have been proposed and studied within discrete fair division, such as EF1 and EFX. Among these, EFX is arguably the strongest and most sought-after notion, however is relatively quite difficult to prove existence for (discussed below).

\paragraph{Envy-freeness up to \textit{one} item (EF1)} This notion was introduced by Budish \cite{Budish2011EF1} in the context of dividing goods. An allocation $X=(X_1,\dots,X_n)$ of goods is said to be EF1 if no agent $i$ envies another agent $j$ after the removal of \textit{some} item from \textit{agent $j$'s bundle}, i.e. $v_i(X_i) \geq v_i(X_j \setminus g)$ for \textit{some} $g \in X_j$. In other words, envy is allowed, but for any pair of agents, if there is an item in one's bundle such that once that item is removed, that agent is no longer envied by the other agent. 
Due to Lipton et. al. \cite{Lipton2004EF1Exists} it follows that an EF1 allocation can always be constructed in polynomial time for arbitrary monotone valuations. In case of additive valuations, Bergman, et al. \cite{BergmanEF1andPO} shows that an EF1 and Pareto optimal of goods can be computed in pseudo-polynomial time, and in polynomial time if every agent has a binary value for every good. 

\paragraph{EF1 for the case of chores.} With chores, EF1 is defined sightly differently: An allocation $X$ of chores is said to be EF1 if no agent $i$ envies another agent $j$ after the removal of \textit{some} chore in \textit{i's own bundle}, i.e. $c_i(X_i \setminus b) \leq c_i(X_j)$ for \textit{some} $b \in X_i$. Just like the case of goods, not only an EF1 allocation of chores always exist and can be computed efficiently \cite{Lipton2004EF1Exists, Bhaskar2020EF1Chores}. The problem of EF1+PO allocation remains open for chores, while 
Garg, et al. \cite{Garg2022EF1POchores} showed that if all cost functions are bi-valued, then in polynomial time, one can construct an EF1 and Pareto optimal allocation. However, an EF1 can be highly undesirable: Intuitively, for the case of chores, EF1 requires envy to disappear after the removal of the \textit{most burdensome chore} according to the envying agent from her own bundle. However, in many cases, such envy exists primarily because of an extremely burdensome chore, so the disappearance of envy after removing that particular chore grants no welfare to the stressful agent who was assigned too much work. A symmetric argument explains the problematic nature of this notion for the case of goods. This issue motivates stronger notions of fairness.

\paragraph{Envy-freeness up to \textit{any} item (EFX) for goods} This notion is a relaxation of envy-freeness but more restrictive than EF1, and it was introduced by Caragiannis, et al. \cite{Caragiannis2016EFX}. An allocation $X$ of goods is said to be EFX if no agent $i$ envies another agent $j$ after the removal of \textit{any} item from \textit{agent $j$'s bundle}, i.e. $v_i(X_i) \geq v_i(X_j \setminus g)$ for \textit{any} $g \in X_j$. In contrast to EF1, the envy between any pair of agents $i$ and $j$ must disappear after the removal of the \textit{least valuable item} according to $i$ from $j$'s bundle. This notion is widely favored in the community of fair division. For instance, Caragiannis, et al. \cite{10.1145/3328526.3329574} remark that
\begin{center}
    \textit{"Arguably, EFX is the best fairness analog of envy-freeness of indivisible items".}
\end{center}
Unlike an EF1 allocation which certainly exists, the existence of complete EFX allocations (all items are distributed) for both goods and chores is unknown. As Caragiannis, et al. commented in \cite{Caragiannis2016EFX}, for the case of goods,
\begin{center}
    \textit{"Despite significant effort, we were not able to settle the question of whether an EFX allocation always exists (assuming all goods must be allocated), and leave it as an enigmatic open question".}
\end{center}

Nevertheless, for the case of goods, some impressive progress has been made in recent years, and it is known that a complete EFX allocation always exists when:
\begin{itemize}
    \item There are two agents, or $n$ agents with an identical valuation \cite{PR2020EFX}.
    \item There are three agents with additive valuations \cite{Existencefor3}, or further, with nice-cancelable valuations, which generalizes additive valuations \cite{Berger_Cohen_Feldman_Fiat_2022}.
    \item There are three agents with two general monotone valuations and one {\em MMS-feasible} valuation, which generalizes nice-cancelable valuations.
    \cite{Akrami2022Simplification}.
    \item There are $n$ agents with additive valuations, but each valuation has at most two different values for goods \cite{10.5555/3491440.3491444}.
    \item There are $n$ agents with at most two additive valuations \cite{arxiv.2008.08798}.
\end{itemize}
Even for the case when there are four agents and additive valuations, the existence of an EFX allocation remains open. 

\paragraph{Envy-freeness up to \textit{any} item (EFX) for chores} The case of chores remain even less understood than the case of goods. Like its analog for goods, an allocation $X$ of chores is EFX if no agent $i$ envies another agent $j$ after the removal of \textit{any} chore in \textit{$i$'s own bundle}, i.e. $c_i(X_i \setminus b) \leq c_i(X_j)$ for any $b \in X_i$. In other words, for any agent $i$, this definitions requires $i$ not to envy anyone else after removing the even least burdensome chore from $i$'s bundle. Intuitively, this suggests that for every agent, her bundle of assignments is not significantly more burdensome than others'. Clearly, when there are only two agents with general valuations, an EFX allocation can always be constructed via the cut-and-choose protocol. In 2020, Chen and Liu \cite{arxiv.2005.04864} proved the existence of an EFX allocation of a set containing both goods and chores for $n$ agents with an identical valuation and cost function, which can be constructed using a protocol modified from \cite{PR2020EFX}. Very recently, Li, et al. \cite{Li2022IDOforChores} showed that an EFX allocation always exists when there are $n$ agents whose cost functions share an identical ordering of chores, and Gafni, et al. \cite{arxiv.2109.08671} showed a positive result when there are $n$ agents and for each cost function, a set of chores with a higher cardinality is always more burdensome than a set with a lower cardinality. Nonetheless, even for three agents with additive cost functions, the existence of a complete EFX allocation is still an open problem.

\subsection{Our Contributions}

We study the existence of EFX allocation of chores for three agents when two of them have additive cost functions with further restrictions, while the third agent has general monotone cost function. 

\paragraph{Existence of an EFX allocation under constraints} If two out of the three agents functions share an identical ordering on chores, are additive and evaluates every non-singleton set of chores as more burdensome than any single chore, then a complete EFX allocation exists.

\paragraph{Existence of a tEFX allocation under relaxed constraints} tEFX stands for {\em transfer EFX} and is a slightly weaker notion than EFX. An allocation $X$ said to be tEFX, if the envy of any agent $i$ disappears for any other agent $j$ after {\em transfer} of one chore from $X_i$ to $X_j$. That is, for any $i,j\in N$, $c_i(X_i\setminus \{b\}) \le c_i(X_j\cup \{b\}),\ \forall g\in X_i$. We show existence of tEFX when two out of the three agents have additive cost functions, and for each of them, the cost of her most burdensome chore is at most double of her least burdensome chore, that is $\frac{\max_{b \in M} c_i(\{b\})}{\min_{b \in M} c_i(\{b\})}\le 2$ for $i=1,2$. We note that, if all of the agents have additive cost functions with max-to-min ratio bounded by $2$, {\em i.e.,} $\frac{\max_{c\in M} c_i(\{b\})}{\min_{b \in M} c_i(\{b\})}\le 2,\ \forall i\in N$, then any EF1 allocation would also be a tEFX allocation. However, this does not subsume our result since our third agent can have arbitrary monotone cost function. 

Our approach is mainly inspired by the approach of \cite{Akrami2022Simplification} that moves in the complete allocation space while maintaining a potential function that defined using is values of the {\em primary} agent for her {\em EFX feasible} bundles. 

\section{Preliminaries}

Let $N=\{1,2,3\}$ be a set of three agents and $M$ be a set of $m$ indivisible chores. Each agent $i \in N$ has a cost function $c_i: 2^M \rightarrow \mathbb{R}_{\geq0}$. We assume that:
\begin{itemize}
    \item Every cost function $c_i$ is normalized, i.e. $c_i(\emptyset)=0$;
    \item It is monotone: $S \subseteq T$ implies $c_i(S) \leq c_i(T)$ for any $S, T \subseteq M$;
    \item In addition to the pervious assumption and without loss of generality, $c_1$ and $c_2$ are additive: for any $S \subseteq M$, $c_j(S) = \sum_{b \in S} c_j(b)$ for $j \in \{1,2\}$.
\end{itemize}


To simplify the notations, we write $c_i(b) = c_i(\{b\})$ for every $b \in M$, and $S \sim_i T$ as $c_i(S) \sim c_i(T)$ with $\sim \in \left\{ \leq, \geq, <, > \right\}$. For any $S, T \subseteq M$, we say that $S$ is more costly or more burdensome for $i$ than $T$ if $S >_i T$. Given an allocation $X = \left( X_1, X_2, \cdots, X_n \right)$, $X_i$ is a \textit{bundle} given to the agent $i$, and we say that $i$ \textit{strongly envies} $j$ if there is a chore $b$ in $X_i$, such that $c_i(X_i \setminus b) > c_i(X_j)$. By definition, $X$ is EFX if no agent strongly envies another. For any pair of agents $i$ and $j$, the bundle $X_i$ is said to be \textit{EFX-feasible} for agent $j$ if, once owning $X_i$, agent $j$ does not strongly envy anyone else. By definition, an allocation $X=(X_1,X_2,X_3)$ is an EFX allocation if for every agent $i$, the bundle given to $i$ is EFX-feasible for herself. As discussed earlier, although EFX is widely favored, it might be too restrictive, and this motivates the study of a slightly relaxed notion of envy-freeness:

\begin{definition} \label{def:tEFX}
An allocation $X = \left( X_1, X_2, \cdots, X_n \right)$ is said to be envy-free up to \textit{transferring} any chore (tEFX) if for all $i, j \in N$ and $b \in M$,
\begin{equation*}
    c_i (X_i \setminus b) \leq c_i (X_j \cup b).
\end{equation*}
\end{definition}
In other words, this relaxed definition requires that no agent $i$ envies another agent $j$ after the \textit{transfer} of \textit{any} chore from $i$'s bundle to $j$'s bundle.

In the following of this section, we explain some technical ideas and notions used in some earlier related papers that will be used again here.

\paragraph{Non-degenerate instances:} This notion was introduced by Chaudhury, et al. \cite{Existencefor3} for the case of goods, and interested readers may refer the corresponding definition there. Over here, we define the related concepts similarly, but for the case of chores. Let $\mathcal{C} = \left\{ c_1, \cdots, c_n \right\}$ be the set of cost functions for $n$ agents and $M$ be the set of $m$ chores. We call a triple $I = \langle [n], M, \mathcal{C} \rangle$ an \textit{instance}. An instance $I$ is non-degenerate if and only if no agent evaluates two different sets equally, i.e. for any $i \in [n]$ and $S, T \subseteq M$,
\begin{equation*}
    S \neq T \Longrightarrow c_i(S) \neq c_i(T).
\end{equation*}


In \cite{Existencefor3}, the authors proved a result implying the sufficiency to deal with only non-degenerate instances for the case of goods with three additive valuations. Later, Akrami, et al. \cite{Akrami2022Simplification} extended the same result to the case with $n$ agents with general monotone valuations. \Cref{lemma:nondeg} below shows that the same also holds for the case of chores with general monotone cost functions, and the proof is similar with that in \cite{Akrami2022Simplification}.

Before formally stating and proving \Cref{lemma:nondeg}, we first perturb the original instance $I$ to construct another one. Write $M = \left\{ b_1, \cdots, b_m \right\}$, and for any $\varepsilon>0$, we perturb an instance $I$ to $I_{\varepsilon} = \langle [n], M, \mathcal{C}_{\varepsilon} \rangle$, where for every $c_i \in \mathcal{C}$, we define $c'_i \in \mathcal{C}_{\varepsilon}$ as
\begin{equation*}
    c'_i(S) = c_i(S) + \varepsilon \cdot \left( \sum_{b_j \in S} 2^j \right), \quad \forall S \subseteq M.
\end{equation*}

\begin{lemma} \label{lemma:nondeg}
Given $n$ agents with monotone cost functions $c_1, \cdots, c_n$, define the perturbed instance $I_{\varepsilon}$ as above, and define the quantity
\begin{equation*}
    \delta = \min_{i \in [n]} \min_{S,T: c_i(S) \neq c_i(T)} \abs{c_i(S) - c_i(T)},
\end{equation*}
and let $\varepsilon>0$ be such that $\varepsilon \cdot 2^{m+1} < \delta$. Then,
\begin{itemize}
    \item For any agent $i$ and $S, T \subseteq M$, $c_i(S) > c_i(T)$ implies $c'_i(S) > c'_i(T)$.
    \item $I_{\varepsilon}$ is a non-degenerate instance.
    \item If $X = (X_1, X_2, X_3)$ is an EFX (or tEFX) allocation for $I_{\varepsilon}$, then it is also an EFX (or tEFX) allocation for $I$.
\end{itemize}
\end{lemma}

\begin{proof}
For the first statement, suppose $S >_i T$, then we have
\begin{align*}
    c'_i(S) - c'_i(T) & = c_i(S) - c_i(T) + \varepsilon \cdot \left( \sum_{b_j \in S \setminus T} 2^j - \sum_{b_j \in T \setminus S} 2^j \right)\\
    & \geq \delta - \varepsilon \cdot \left( \sum_{b_j \in T \setminus S} 2^j \right)\\
    & \geq \delta - \varepsilon \cdot (2^{m+1} - 1) > 0.
\end{align*}
For the second statement, consider any two sets $S, T \subseteq M$ such that $S \neq T$. If $c_i(S) \neq c_i(T)$, then by the first statement, we have $c'_i(S) \neq c'_i(T)$. If $c_i(S) = c_i(T)$, then
\begin{equation*}
    c'_i(S) - c'_i(T) = \varepsilon \cdot \left( \sum_{b_j \in S \setminus T} 2^j - \sum_{b_j \in T \setminus S} 2^j \right) \neq 0,
\end{equation*}
because $S \neq T$ so their elements have different indices. Therefore, $I_{\varepsilon}$ is non-degenerate.

For the last statement, suppose $X$ is an EFX allocation in $I_{\varepsilon}$ but not EFX in $I$, then there exists a pair of agents $i$ and $j$, and $b \in X_i$ such that $c_i(X_i \setminus b) > c_i(X_j)$. By the first statement, this implies that $c'_i(X_i \setminus b) > c'_i(X_j)$ so it is not EFX in $I_{\varepsilon}$, which is a contradiction. The proof for the tEFX case is the same, since $c_i(X_i \setminus b) > c_i(X_j \cup b)$ also implies $c'_i(X_i \setminus b) > c'_i(X_j \cup b)$.
\end{proof}

Thanks to this result, from now on we can only consider non-degenerate instances. Consequently, every chore has a positive and distinct cost under every cost function.

Next, we introduce some definitions on cost functions, which are imperative in later chapters.

\begin{definition}
We define the following properties of cost functions.
\begin{itemize}
    \item A cost function $c: 2^M \rightarrow \mathbb{R}_{\geq 0}$ is \textit{collective} if for any $S \subseteq M$ such that $S$ contains at least two chores, then $c(S) > c(b)$ for any $b \in M$.
    \item For any $\alpha>0$, a cost function $c$ is \textit{$\alpha$-ratio-bounded} if
    \begin{equation*}
        \frac{\max_{b \in M} c(b)}{\min_{b \in M} c(b)} \leq \alpha.
    \end{equation*}
    \item Two cost functions $c$ and $c'$ are \textit{IDO} if they share an \textit{identical ordering} among the costs of chores, i.e. for any $b, b' \in M$, $c(b) > c(b')$ if and only if $c'(b) > c'(b')$.
\end{itemize}
\end{definition}

One may easily observe that for an additive cost function, if it is 2-ratio-bounded, then it must be collective, but not conversely: Consider a set $M = \{b_1, b_2, b_3\}$ of three chores and an additive cost function $c$ such that $c(b_1)=0.2$, $c(b_2)=0.3$, and $c(b_3)=0.49$; clearly $c$ is collective, but $c(b_3)/c(b_1)>2$.

\section{Technical Overview}

In this section, we briefly illustrate the main technical insights we use to show our results, which are mainly inspired by those in \cite{Akrami2022Simplification}. We present an algorithmic proof for the existence of (i) an EFX allocation and (ii) a tEFX allocation under different assumptions on two of the three cost functions. Through our entire framework, the third agent is unrestricted besides having a monotone cost function. 
Similar with Akrami, et al. \cite{Akrami2022Simplification}, for the EFX case, we rely on improving a potential function in the space of \textit{complete} \textbf{non}-EFX allocations, instead of moving in the space of \textit{partial} EFX allocations like in \cite{Existencefor3} and \cite{Chaudhury2021Charity}. For the tEFX case, our iterating steps are similar, but a potential function is not needed because the algorithm may only complete at most $m$ iterations.

The central idea, which will be applied on the EFX case, is to decrease the output of a positive potential function while iterating over complete non-EFX allocations of the chores satisfying certain invariant. Because they are complete, if during this process we obtain an EFX allocation, then the algorithm stops and we are done. Otherwise, the iteration continues and the potential keeps going down. Be aware that there are finitely many possible allocations because our task is to allocate finitely many chores to finitely many agents, and thus eventually the potential can no longer be decreased. Readers may notice from the process that during every iteration, the more costly bundle between $X_1$ and $X_2$ according to 1 must be relaxed in terms of workload, for the purpose of improving the potential. 

More specifically, for both the EFX and tEFX cases, if a (t)EFX allocation has not been found, then we keep constructing an allocation $X = (X_1, X_2, X_3)$ that satisfies the following invariant:
\begin{enumerate}
    \item Agent 1 finds both $X_1$ and $X_2$ (t)EFX-feasible.
    \item Agent 2 finds $X_3$ (t)EFX-feasible.
\end{enumerate}
Note that such an allocation always exists: as elaborated in \cite{Li2022IDOforChores}, we may always find an EFX allocation with respect to agent 1's cost function $c_1$ alone. Next, without loss of generality, assume $X_3$ is agent 2's favorite bundle, i.e. $c_2(X_3) = \min_{j \in N} c_2(X_j)$, which is certainly (t)EFX-feasible for her, and hence both properties hold. The key insight of the iterating steps is that, after each update on $X$, the new allocation still satisfies the invariant.

The following arguments are symmetric to those in \cite{Akrami2022Simplification}. If $X_3$ is not the only (t)EFX-feasible bundle for both 2 and 3, i.e. if any of agents 2 and 3 finds at least one of $X_1$ or $X_2$ (t)EFX-feasible, then $X$ is a (t)EFX allocation. Again, suppose $X_3$ is (t)EFX-feasible for agent 2 without loss of generality, and we do a simple case analysis:
\begin{itemize}
    \item If agent 2 finds $X_1$ or $X_2$ (t)EFX-feasible, then we let agent 3 pick her favorite bundle among all three, which is certainly (t)EFX-feasible for her. If she picks $X_3$, then we give $X_1$ or $X_2$ to agent 2, whichever (t)-EFX feasible for her; if agent 3 picks $X_1$ or $X_2$, then we give $X_3$ to agent 2. Recall that all three bundles are (t)EFX-feasible for agent 1, so in both cases she can take the last bundle.
    \item If agent 3 finds $X_1$ or $X_2$ (t)EFX-feasible, then we may let her pick that, agent 2 pick $X_3$, and agent 1 pick the remaining one.
\end{itemize}
Therefore, the only unresolved scenario is when $X_3$ is the only (t)EFX-feasible bundle for both agents 2 and 3. Intuitively, this implies that 2 and 3 find both $X_1$ and $X_2$ "too burdensome". 

To summarize our approach to solve the EFX case in non-technical terms, we will assume agent 1 evaluates $X_1$ more costly than $X_2$, remove some chores from $X_1$, and allocate them to $X_3$ to \textit{balance} the costs between $X_3$ and the $X_1$. Next, based upon agent 1's cost function $c_1$, we may choose to augment the relieved $X_1$ with 1's least costly chore in $X_2$, depending on which bundle, between $X_1$ and $X_2$, that 1's least costly chore among all chores is located at. We will show that, after each iteration, thanks to the assumptions on $c_1$ and $c_2$, this protocol guarantees the decrease of potential with the new bundle, which also satisfies the invariant.

For the tEFX case, the similar approach will be applied, and we will show the protocol outputs a tEFX allocation in polynomial time, even if we overall significantly relax the assumptions on $c_1$ and $c_2$ by removing the IDO assumption and slightly strengthening the collectivity assumption.

More techniques in \cite{Akrami2022Simplification} served to prove a strong result on the existence of EFX allocations for goods, but many of them are, to our best knowledge, are not directly applicable for the case of chores, just like those in \cite{Existencefor3}. 

\section{EFX Existence for Two IDO Cost Functions}

\begin{algorithm}[t] \label{algoinCh4}
	\caption{Finding an EFX Allocation}
	\begin{algorithmic}[1]
	\Require 
	    \begin{enumerate}
	        \item Three cost functions $c_1, c_2, c_3$ that satisfy Assumptions \ref{assum:additive} - \ref{assum:ido}.
	        \item The set $M$ containing $m$ chores.
	        \item An allocation $X=(X_1,X_2,X_3)$ of $M$ such that:
	        \begin{itemize}
	            \item $X_1$ and $X_2$ are EFX-feasible for agent 1;
	            \item $X_3$ is EFX-feasible for agent 2.
	        \end{itemize}
	        \item The costs of certain bundles for certain agents; in particular we need: $c_1(X_1), c_1(X_2), c_2(X_3)$.
	        \item The costs of all chores according to agents 1 and 2, i.e. the set of costs $\{c_i(b_j)\}$ for $i \in N$, $j \in [m]$.
	    \end{enumerate}
	\Ensure An EFX allocation $Y=(Y_1,Y_2,Y_3)$
		\While {$X$ is not an EFX allocation}
		    \State Label agent 1's more burdensome bundle between $X_1$ and $X_2$ as $X_1$, and the other as $X_2$.
		    \State Label agent 1's least costly chore in $X_1$ as $d$, and her least costly chore in $X_2$ as $d'$.
		    \If {According to 1, the least costly chore $b$ in $X_3$ is less costly than $d$ according to 1}
		        \State Transfer $d$ from $X_1$ to $X_3$, and $b$ from $X_3$ to $X_1$
		        \State $X_1 = \left( X_1 \setminus d \right) \cup b$, $X_3 = \left( X_3 \cup d \right) \setminus b$
		    \Else
		        \If {$d$ is agent 1's least costly chore in $M$}
		            \State Transfer $d$ from $X_1$ to $X_3$.
		            \State $X_1 = X_1 \setminus d$, $X_3 = X_3 \cup d$.
		        \ElsIf {$d'$ is agent 1's least costly chore in $M$}
		            \State First transfer $d$ from $X_1$ to $X_3$, and then transfer $d'$ from $X_2$ to $X_1$.
		            \State $X_1 = \left( X_1 \setminus d \right) \cup d'$, $X_2 = X_2 \setminus d'$, $X_3 = X_3 \cup d$.
		        \EndIf
		    \EndIf
		\EndWhile
		\Return the final allocation denoted as $Y = (Y_1, Y_2, Y_3)=(X_1,X_2,X_3)$ when the loop terminates.
	\end{algorithmic}
\end{algorithm}

Specifically, our Algorithm 1 iterates over a subclass of complete allocations of the set of chores $M$ over the three agents. The subclass only contains allocations with the following two properties:
\begin{itemize}
    \item $X_1$ and $X_2$ are EFX-feasible for agent 1.
    \item $X_3$ is EFX-feasible for agent 2.
\end{itemize}
Such a partition certainly exists: We can produce a partition $X = (X_1, X_2, X_3)$ such that all bundles are EFX-feasible for agent 1. Then, agent 2 picks her favorite bundle out of the three, say $X_3$, and we label the other two bundles as $X_1$ and $X_2$; clearly, the invariant holds. Note that, under the invariant, we may assume that $X_3$ is the only EFX-feasible bundle for agents 2 and 3 without loss of generality. If $X_1$ or $X_2$ is EFX-feasible for agent 2, then we may let agent 3 pick her favorite bundle; if that bundle is $X_1$ or $X_2$, then let agent 2 pick $X_3$ and agent 1 pick the last bundle; on the other hand, if agent 3's favorite is $X_3$, then let agent 2 pick the other favorite bundle of hers and let agent 1 pick the last one. Likewise, if agent 3 finds at least one of $X_1$ and $X_2$ EFX-feasible, then let her that bundle, agent 2 pick $X_3$ and agent 1 pick the last one. In either case, $X$ is already EFX and we are done.

For a partition $X$ labelled in the method above, we define our potential function as $\phi(X) = \max \left\{ c_1(X_1), c_1(X_2) \right\}$. The remaining part of this section focuses on the process of modifying $X$ and decreasing the potential when $X$ is not an EFX allocation. 

Before presenting more technical results, we first make a few simple observations on the status of $X_3$ for agents 2 and 3, and they hold without more assumptions on $c_1$ and $c_2$ which will be presented later in this section. For any agent $i \in N$, label $b_{i,1}$ as the least costly chore in $X_1$ according to $i$'s evaluation. Then,

\begin{lemma} \label{obs:3isminfor2and3}
If $X$ satisfies the invariant and $X_3$ is the only EFX-feasible bundle for both agents 2 and 3, then for $i \in \{2,3\}$, $c_i(X_3) = \min \left\{ c_i(X_1), c_i(X_2), c_i(X_3) \right\}$.
\end{lemma}

\begin{proof}
We assumed that neither agent 2 nor 3 finds any of $X_1$ and $X_2$ EFX-feasible. Nevertheless, every agent must always have at least one EFX-feasible bundle, including her favorite bundle. If the statement does not hold, i.e. for at least one $i \in \{2,3\}$, there is a $j \in \{1,2\}$ such that $c_i(X_j) < c_i(X_3)$, implying that $X_j$ is EFX feasible for agent $i$, which contradicts the assumption.
\end{proof}

In fact, in this scenario, $X_3$ is by far the most favorable choice for both agents 2 and 3. The next observation further strengthens this claim beyond \Cref{obs:3isminfor2and3}.

\begin{lemma} \label{obs:burdenafterremoval}
If $X$ satisfies the invariant and $X_3$ is the only EFX-feasible bundle for both agents 2 and 3, then for $i \in \{2,3\}$, $X_3 <_i \min_i \left\{ X_1 \setminus b_{i,1}, X_2 \setminus b_{i,2} \right\}$. Equivalently, if $c_i$ is also additive, then $X_3 \cup b_{i,j} <_i X_j$.
\end{lemma}

\begin{proof}
As shown in \Cref{obs:3isminfor2and3}, $X_3$ must be 2 and 3's favorite bundle. Clearly, if $c_i$ is additive and the first statement holds, then the second statement is trivial. We prove the first statement when $i=2$; the proof for $i=3$ is symmetric. Let us assume otherwise and without loss of generality, suppose $X_2 \setminus b_{3,2} <_2 X_1 \setminus b_{3,1}$, then
\begin{equation}
    X_1 >_2 X_3 >_2 \min_2 \left\{ X_1 \setminus b_{3,1}, X_2 \setminus b_{3,2} \right\} =_2 X_2 \setminus b_{3,2}.
\end{equation}
Then by definition of $b_{3,2}$, $X_2$ is EFX-feasible for agent 3 and we obtain a contradiction.
\end{proof}

Next, to show the existence of EFX allocations, we unfortunately have to implement more assumptions on two of the three cost functions. Without loss of generality, we add the following three restrictions on $c_1$ and $c_2$.

\begin{assumption} \label{assum:additive} (Additivity).
$c_1$ and $c_2$ are additive, while $c_3$ can still be any monotone cost function.
\end{assumption}

\begin{assumption} \label{assum:collective} (Collectivity).
For $i \in \{1,2\}$, $c_i(M') > \displaystyle \max_{b \in M} c_i(b)$ for any subset $M' \subseteq M$ containing at least two chores.
\end{assumption}

\begin{assumption} \label{assum:ido} (Identical ordering / IDO).
$c_1$ and $c_2$ share an identical ordering on individual chores: $b >_1 b'$ if and only if $b >_2 b'$.
\end{assumption}


\begin{remark}
Note that both \Cref{obs:3isminfor2and3} and \Cref{obs:burdenafterremoval} hold as long as all three cost functions are monotone; they do not even have to be additive. We will see shortly that our proofs strictly require that two of the three cost functions satisfy all three assumptions concurrently. Therefore, assuming additivity on the third agent is unnecessary.
\end{remark}

Because of \Cref{assum:ido}, agent 1 and agent 2 share a common least costly chore in every bundle and entire $M$. Write $d = \argmin_{x \in X_1} c_{1,2}(x)$ as their least costly chore in $X_1$. Without loss of generality, we assume $X_1 >_1 X_2$, and consider the partition $\hat{X} = \{ X_1 \setminus d, X_2, X_3 \cup d\}$, produced by transferring $d$ from $X_1$ to $X_3$. Furthermore, observe that $\phi(\hat{X}) = \max_1 \left\{ X_1 \setminus d, X_2 \right\} < c_1(X_1) = \phi(X)$.

This entire section is devoted to verify the following statement:
\begin{theorem}
Given three agents $N$, their cost functions $\mathcal{C} = \{ c_1, c_2, c_3 \}$, and a finite set of chores $M$, if $c_1$ and $c_2$ satisfy Assumptions \ref{assum:additive}-\ref{assum:ido}, then an EFX allocation exists.
\end{theorem}

We prove this theorem by dividing into two cases with respect to the costs of items in $X_3$. In particular, we consider the case if there is a chore in $X_3$ that is less costly than $d = \argmin_{x \in X_1} c_{1,2}(x)$ according to agents 1 and 2, and the case otherwise. Note that because of \Cref{assum:ido} and the assumption of non-degeneracy, a chore must be either strictly less or strictly more costly than $d$ for both agents 1 and 2 concurrently.

\subsection*{Case 1: there exists a $b \in X_3$ such that $b <_{1,2} d$}

This is the simpler case and a desirable allocation can be achieved in one step. Consider the partition $X'$ where $X'_1 = \left( X_1 \setminus d \right) \cup b$, $X'_2 = X_2$, and $X'_3 = \left( X_3 \cup d \right) \setminus b$.

\begin{proposition}
$X'$ satisfies the invariant, i.e. $X'_1$ and $X'_2$ are EFX-feasible for agent 1, and $X'_3$ is EFX-feasible for agent 2.
\end{proposition}

\begin{proof}
Clearly, $X'_1 <_1 X_1$, $X'_3 >_1 X_3$, and $X'_2=X_2$, so $X'_1$ must be EFX-feasible for agent 1. If $X'_2 = X_2$ is not EFX-feasible for agent 1, i.e. there exists a chore $f \in X_2$ such that $X_2 \setminus f >_1 X'_1$, then we have
\begin{equation}
    c_1 (X'_2 \setminus f) = c_1(X_2) - c_1(f) > c_1(X_1) - c_1(d) + c_1(b);
\end{equation}
hence,
\begin{equation}
    c_1(X_2) > c_1(X_1) + \underbrace{c_1(f) + c_1(b) - c_1(d)}_{>0} \geq c_1(X_1),
\end{equation}
and this is a contradiction of the assumption $X_1 >_1 X_2$.

For the second claim, observe that for every $g \in X'_3$, we have
\begin{align}
\begin{split}
    c_2(X'_3 \setminus g) & = c_2(X_3) + \underbrace{c_2(d) - c_2(b) - c_2(g)}_{<0}\\
    & < c_2(X_3) < \min \left\{ c_2(X'_1), c_2(X'_2) \right\},
\end{split}
\end{align}
where the first inequality holds because of \Cref{assum:collective} and the second inequality follows from \Cref{obs:burdenafterremoval}.
\end{proof}
Finally, observe that $\phi(X') = \max_1 \left\{ X'_1, X'_2 \right\} < \phi(X)$ because $\phi(X) = c_1(X_1) > c_1(X_2)$ and $X'_1 <_1 X_1$. Therefore, after this one-step operation, we decrease the potential using $X'$ which satisfies the invariant.


\subsection*{Case 2: for every $b \in X_3$, $b >_{1,2} d$}

Note that we may assume $c_i(b) \neq c_i(d)$ for every $i \in N$ and every $b \in X_3$, because of the assumption of non-degeneracy. Following the assumption that $b >_{1,2} d$ for every $b \in X_3$, it is clear that the unique least costly chore in $M$ for both agents 1 and 2 is not in $X_3$. We divide this case into two sub-cases, depending on the location of this least costly chore, whether it lies in $X_1$ or $X_2$. In both cases, the least costly chore $d' = \argmin_{x \in X_2} c_{1,2}(x)$ for agents 1 and 2 in $X_2$ will appear, as illustrated in Algorithm 1.

\subsubsection*{Case 2.1: the minimum chore lies in $X_1$}
In this case, $d$ is the least costly chore for agents 1 and 2 in $M$ and thus $d' >_{1,2} d$. we claim that the partition $\left \{X_1 \setminus d, X_2, X_3 \cup d \right\}$ is enough by proving it satisfies the invariant. Indeed, for every $f \in X_2$, we have
\begin{equation}
    X_2 \setminus f \leq_1 X_2 \setminus d' <_1 X_1 \setminus d' <_1 X_1 \setminus d.
\end{equation}
Moreover, for any $e \in X_3 \cup d$, $(X_3 \cup d) \setminus e \leq_2 X_3 <_2 \min_2 \left\{X_1 \setminus d, X_2 \right\}$, where the last inequality holds because of \Cref{obs:3isminfor2and3} and \Cref{obs:burdenafterremoval} by setting $i=2$ and $j=1$. Apparently, as we assumed $X_1$ is EFX-feasible for agent 1 within $X$, we have $X_1 \setminus d <_1 X_2$, so the potential of this new allocation is $c_1(X_2)$, which is assumed to be lower than $\phi(X) = c_1(X_1)$. We will then analyze the case if the minimum chore lies in $X_2$.

\subsubsection*{Case 2.2: the minimum chore lies in $X_2$}
Clearly, in this case, $d'$ is the unique least costly chore for agents 1 and 2 among all chores in $M$. Once again, within the partition $\hat{X} = \left \{X_1 \setminus d, X_2, X_3 \cup d \right\}$, $X_1 \setminus d$ is apparently EFX-feasible for agent 1. We claim that $X_3 \cup d$ is EFX-feasible for agent 2. Indeed, for every $e \in X_3 \cup d$, we have $(X_3 \cup d) \setminus e <_2 X_3 <_2 \min_2 \left\{ X_1 \setminus d, X_2 \right\}$, thanks to \Cref{obs:burdenafterremoval}. If $X_2$ is EFX-feasible for agent 1, we are done and start the next round. Suppose not, then note that the issue is irrelevant with $X_3 \cup d$: $X_2$ was EFX-feasible for agent 1 within $X$, so for any $e \in X_2$, $X_2 \setminus e <_1 X_3 <_1 X_3 \cup d$. Thus, if $X_2$ is not EFX-feasible for agent 1 within $\hat{X}$, we must have $X_2 \setminus d' >_1 X_1 \setminus d$. To fix this, consider the allocation $X' = \left\{ \left( X_1 \setminus d \right) \cup d', X_2 \setminus d', X_3 \cup d \right\}$, and we claim that this partition satisfies the invariant.

Note that, in $X'$, clearly $d' = \argmin_{x \in X'_1} c_{1,2}(x)$, and therefore it is the least costly chore in the entire $M$ according to both agents 1 and 2. As a result, $X'_1 \setminus d' = X_1 \setminus d <_1 X_2 \setminus d'$, where the last inequality holds because we assumed $X_2 = \hat{X}_2$ is not EFX-feasible within $\hat{X}$ and $d' = \argmin_{x \in \hat{X}_2} c_{1,2}(x)$. Therefore, $X'_1$ is EFX-feasible for agent 1 within $X'$. For any $f \in X_2 \setminus d'$, observe that $X'_2 \setminus f = X_2 - d' - f <_1 X_2 \setminus d <_1 X_1 \setminus d <_1 X'_1$. Finally, for every $e \in X_3 \cup d$, since we assumed $e \geq_{1,2} d$, we have
\begin{align}
\begin{split}
    c_2(X'_3 \setminus e) \leq c_2(X_3) & < \min \left\{ c_2(X_1 \setminus d), c_2(X_2 \setminus d') \right\}\\
    & \leq \min \left\{ c_2(X'_1), c_2(X'_2) \right\}
\end{split}
\end{align}
so $X'_3$ is EFX-feasible for agent 2.

Finally, easily observe that $X'_1 <_1 X_1$ and $X'_2 <_1 X_2$, so $\phi(X') < \phi(X)$.

We have covered all possible cases, and conclude that we may find an allocation that satisfies the invariant and reduces the potential under all scenarios. Therefore, because the potential function is bounded below, the algorithm will certainly stop eventually when the potential can no longer decrease and return an EFX allocation. However, recall that $X_1$ may both receive and send chores, so even if the potential decreases at every round, the quantity of change may be very small. Consequently, this algorithm is pseudo-polynomial time only.

\section{tEFX Existence for Two 2-ratio-bounded Cost Functions}

\begin{algorithm}[t] \label{algoinCh4}
	\caption{Finding a tEFX Allocation}
	\begin{algorithmic}[1]
	\Require 
	    \begin{enumerate}
	        \item Three cost functions $c_1, c_2, c_3$ that satisfy Assumptions \ref{assum:additive} and \ref{assum:c1c2ratiobound}.
	        \item The set $M$ containing $m$ chores.
	        \item An allocation $X=(X_1,X_2,X_3)$ of $M$ such that:
	        \begin{itemize}
	            \item $X_1$ and $X_2$ are EFX-feasible (thus tEFX-feasible) for agent 1;
	            \item $X_3$ is EFX-feasible (thus tEFX-feasible) for agent 2.
	        \end{itemize}
	        \item The costs of certain bundles for certain agents; in particular we need: $c_1(X_1), c_1(X_2), c_2(X_3)$.
	        \item The costs of all chores according to agents 1 and 2, i.e. the set of costs $\{c_i(b_j)\}$ for $i \in N$, $j \in [m]$.
	    \end{enumerate}
	\Ensure A tEFX allocation $Y=(Y_1,Y_2,Y_3)$
		\While {$X_3$ is still agent 2's favorite bundle}
		    \State Label agent 1's more burdensome bundle between $X_1$ and $X_2$ as $X_1$, and the other as $X_2$.
		    \State Label agent 2's least costly bundle in $X_1$ as $d'$.
		    \State Transfer $d'$ from $X_1$ to $X_3$.
		    \State $X_1 = X_1 \setminus d'$, $X_3 = X_3 \cup d'$.
		\EndWhile
		\Return the final allocation denoted as $Y = (Y_1, Y_2, Y_3)=(X_1,X_2,X_3)$ when the loop terminates.
	\end{algorithmic} 
\end{algorithm}

In this section, we aim to find an allocation $\{ X_1, X_2, X_3 \}$ such that every $X_i$ is \textbf{tEFX-feasible} for agent $i$, where the definition is given in \Cref{def:tEFX}.

All three cost functions are still assumed to be monotone, and further $c_1$ and $c_2$ must be additive. However, instead of collectivity, we will assume $c_1$ and $c_2$ are 2-ratio-bounded, which is slightly more restrictive than collectivity, as explained at the end of Section 2. The new assumption is formally phrased as:

\begin{assumption} \label{assum:c1c2ratiobound}
For $i \in \{1,2\}$, $\frac{\max_{b \in M} c_i(b)}{\min_{b \in M} c_i(b)} \leq 2$.
\end{assumption}

Once again, our proof requires two of the three functions to be additive and 2-ratio-bounded concurrently, so it is sufficient to relax the third agent.
%
In short, the only assumptions in this section are Assumptions \ref{assum:additive} and \ref{assum:c1c2ratiobound} (no IDO assumption).

From previous arguments, we may generate an allocation $X=(X_1, X_2, X_3)$ such that every bundle is EFX-feasible (not only tEFX-feasible) for agent 1. Then let $X_3$ be the least costly bundle for agent 2, and label the remaining two bundles as $X_1$ and $X_2$. Clearly this allocation satisfy the following invariant which will be maintained during the protocol.
\begin{enumerate}
    \item $X_1$ and $X_2$ are tEFX-feasible for agent 1; and
    \item $X_3$ is the most favorable bundle for agent 2, i.e. $X_3 <_2 \min_2 \left\{ X_1, X_2 \right\}$.
\end{enumerate}
Like in the previous section, suppose the invariant holds, then without loss of generality, we may assume the following:
\begin{itemize}
    \item $X_1 >_1 X_2$;
    \item $X_3$ is the only tEFX-feasible bundle for both agents 2 and 3: Otherwise, we are done by a symmetric argument of the beginning of Section 4.
\end{itemize}

The next lemma follows from the above observation. 

\begin{lemma}
For $i \in \{2,3\}$, if $X$ satisfies the invariant and the assumption that $X_3$ is the only tEFX-feasible bundle for both agents 2 and 3, then $c_i(X_3) = \min \left\{ c_i(X_1), c_i(X_2), c_i(X_3) \right\}$.
\end{lemma}

The goal of this section is to prove the following theorem.

\begin{theorem} \label{thminsec5}
Given a set $N$ of three agents, their cost functions $\mathcal{C} = \{ c_1, c_2, c_3 \}$, and a finite set of chores $M$, if $c_1$ and $c_2$ are additive and satisfy \Cref{assum:c1c2ratiobound}, then a tEFX allocation can be found in polynomial time $O(m)$.
\end{theorem}

To proceed, we need to fix the notations for various "minimum" chores. Let $d = \argmin_{x \in X_1} c_1(x)$, $d' = \argmin_{x \in X_1} c_2(x)$, and $q = \argmin_{x \in X_2} c_1(x)$. By definition, $d' \geq_1 d$.

Consider the partition $X' = \left\{ X_1 \setminus d', X_2, X_3 \cup d' \right\}$, constructed by moving the least costly chore in $X_1$ according to agent 2 to $X_3$. 

\begin{proposition} \label{prop:tEFX-both-cost-ratio}
Let $X=(X_1,X_2,X_3)$ be an allocation where $X_1$ and $X_2$ are tEFX-feasible for agent 1 in $X$, then $X'_1=X_1 \setminus d'$ and $X'_2=X_2$ are tEFX-feasible for agent 1 within $X'$.
\end{proposition}

\begin{proof}
The statement clearly holds for $X_1 \setminus d'$: since $X_1$ itself was tEFX-feasible against $X_2$ and $X_3$, it follows that $X_1\setminus d'$ remains tEFX-feasible against $X_2$ and $X_3 \cup d'$. 

Now we prove $X'_2=X_2$ is tEFX-feasible for agent 1. Assume otherwise, so that $X_2 - q >_1 X_1 - d' + q$, then
\begin{equation}
    c_1(X_2) - c_1(q) > c_1(X_1) - c_1(d') + c_1(q); 
\end{equation}
and therefore,
\begin{equation}
    c_1(X_2) > c_1(X_1) + 2 c_1(q) - c_1(d') \geq c_1(X_1),
\end{equation}
where the last inequality holds because of \Cref{assum:c1c2ratiobound}, and this is a contradiction on the assumption that $X_1 >_1 X_2$.
\end{proof}

The only remaining problem is the tEFX-feasibility of $X'_3$ for agent 2. We analyze this problem via the cost of $X'_3$ from agent 2's perspective.

The first and simplest case is when $X_3 \cup d' <_2 \min_2 \left\{ X_1 \setminus d', X_2 \right\}$, then $X_3 \cup d'$ is indeed tEFX-feasible for agent 2 and we are done because $X'$ satisfies the two invariants. 

Therefore, from now on, we assume $X_3 \cup d' >_2 \min_2 \left\{ X_1 \setminus d', X_2 \right\}$. If $X_1 \setminus d' <_2 X_2$, then we also have $X_1 \setminus d' <_2 X_3 \cup d'$. Then, by the definition of $d'$, it follows that $X_1$ is tEFX-feasible for agent 2, which contradicts one of our earlier assumptions that $X_3$ is the only tEFX-feasible bundle for 2. So, we may safely assume $X_1 \setminus d' >_2 X_2$. Note that, in this case, $X_2$ is the most favorable bundle for agent 2 within $X'$ and hence tEFX-feasible for agent 2. Therefore, if we prove that $X_3 \cup d'$ is also tEFX-feasible for agent 2 in this case, then the algorithm terminates because $X'$ is an tEFX allocation, and our next proposition shows this is indeed the case.

\begin{proposition} \label{prop:X3istEFXwhenX2Lowest}
If $X_2 <_2 X_1 \setminus d'$, then $X_3 \cup d'$ is tEFX-feasible for agent 2 in $X'$.
\end{proposition}

\begin{proof}
Observe the inequality $X_3 <_2 X_2 <_2 X_3 \cup d'$, where the first part holds because $X_3$ is tEFX-feasible for agent 2 but neither $X_1$ nor $X_2$ is within $X$, and the second part holds because $X_3 \cup d' >_2 \min_2 \left\{ X_1 \setminus d', X_2 \right\} = X_2$. Therefore, we may write $c_2(X_2) = c_2(X_3) + \delta \cdot c_2(d')$ for some $\delta \in (0, 1)$. Define $f = \argmin_{x \in X_3 \cup d'} c_2(x)$, and suppose $X_3 \cup d'$ is not tEFX-feasible for agent 2 in $X'$, then we must have
\begin{equation}
    X_3 + d' - f >_2 \min_2 \left\{ X_1 - d' + f, X_2 + f \right\} =_2 X_2 \cup f,
\end{equation}
which implies
\begin{equation}
    c_2(X_3) + c_2(d') - c_2(f) > c_2(X_2) + c_2(f).
\end{equation}
Note that $f \neq d'$: The other two bundles in $X'$ adding $d'$ are $X'_1 \cup d' = X_1$ and $X_2 \cup d'$, which are both more costly than $X'_3 \setminus d' = X_3$. Hence, $f <_2 d'$ and
\begin{align}
    c_2(X_3) & > c_2(X_2) +2c_2(f) - c_2(d')\\
    & = c_2(X_3) + \delta \cdot c_2(d') + 2c_2(f) - c_2(d')\\
    & = c_2(X_3) + 2c_2(f) - (1-\delta) \cdot c_2(d').
\end{align}
Therefore, $2c_2(f) - (1-\delta) \cdot c_2(d')<0$, implying
\begin{equation}
    c_2(f) < \frac{1-\delta}{2} \cdot c_2(d') < \frac{1}{2} \cdot c_2(d') \Longrightarrow \frac{c_2(d')}{c_2(f)} > 2.
\end{equation}
But by \Cref{assum:c1c2ratiobound}, this is not possible.
\end{proof}

We argued earlier that in this case, $X_2$ is tEFX-feasible for agent 2, and \Cref{prop:X3istEFXwhenX2Lowest} shows that $X'_3$ is also tEFX-feasible for agent 2. Together with \Cref{prop:tEFX-both-cost-ratio}, $X'$ is an tEFX allocation by letting agent 3 pick her most favorable bundle. Consequently, the algorithm terminates.

\subsection{Time Complexity of the Algorithm}

We now show that the time complexity of our algorithm in this section is $O(m)$, where $m$ is the total number of chores. The input of the algorithm is an EFX allocation $X = \left( X_1, X_2, X_3 \right)$ with respect to agent 1's cost function, and certainly $X_1 \cup X_2$ contains at most $m$ elements.

Recall that the algorithm in Section 4 involves chore transfer among all three agents: In Case 1, $X_1$ and $X_3$ exchange chores, and in Case 2, $X_3$ always receives a chore, and $X_1$ and $X_2$ may exchange. Moreover, in our analysis, it is difficult to tell when that algorithm terminates, and the algorithm therefore takes pseudo-polynomial time.

In contrast, for the algorithm in Section 5, the only receiving party is $X_3$, which is a fixed bundle, unlike $X_1$ and $X_2$ that may shift positions depending on agent 1's cost on them. In other words, during the entire process one of the three bundles only receive chores from the other two parties without sending anything, and similarly the other two parties only gives out chores. Therefore, the algorithm must terminate no later than all chores from $X_1 \cup X_2$ are transferred, which takes at most $m$ iterations.

Now we analyze the time complexity within each iteration. During each time, the algorithm first selects agent 2's most favorable chore from agent 1's less favorable bundle between $X_1$ and $X_2$. Again, like in the main part, we name this chore as $d'$. Both tasks can be done quickly by previously recording the costs of chores and computes the total costs of the three initial bundles. Next, $d'$ is transferred to $X_3$ from one of $X_1$ or $X_2$, and the protocol computes the costs of $X_1 \setminus d'$, $X_2$, and $X_3 \cup d'$, which again can be done in constant time. As we argued earlier, the entire process takes at most $m$ iterations. Therefore, in addition to the envy-cycle elimination algorithm to find the initial partition, which takes $O(mn^3)$, where $n$ is the number of agents; in our case, $n=3$ so it takes only $O(m)$.

Together with our earlier results in this section, we have shown that when two of the cost functions are both additive and 2-ratio-bounded, then a tEFX allocation of chores exist and can be computed in polynomial time, and thus, \Cref{thminsec5} holds.

\section{Conclusion and Future Work}

In this paper, we have shown that:
\begin{enumerate}
    \item An EFX allocation always exists when there are three agents, and among them, two of them have additive cost functions, and their cost functions are further assumed to be collective and have an identical ordering on the chores; and
    \item A tEFX allocation always exists when there are three agents, and two of them have additive and 2-ratio-bounded cost functions. Moreover, this tEFX allocation can be computed efficiently in time $O(m)$.
\end{enumerate}

We believe our work takes us one step closer to the solution of the more difficult problem, namely whether an EFX allocation always exists with three agents with more general cost functions. 

While the third agent's cost function can be arbitrary monotone function, 
our proofs extensively exploit the assumptions on $c_1$ and $c_2$, and unfortunately do not work for the case when the first and second agents have general additive cost functions. Therefore, a natural next step is to investigate the case with less restrictive first two cost functions. We anticipate the difficulty of the case when all three cost functions are generally additive, so another option might be to add more restrictions on the third agent, relax some assumptions on the first two agents, and explore this case.

Even if we confirmed the existence of an EFX allocation with certain restrictions, the algorithm is pseudo-polynomial, since the decrease of the potential after an iteration may be small. It is also worthwhile to investigate whether a polynomial-time protocol exists under a similar or less restrictive setting.

\section*{Acknowledgments}

We thank Dr. Bhaskar Ray Chaudhury for helpful discussions on fair division and related subjects.

\bibliography{refs}

\end{document}